\documentclass[a4paper,UKenglish,autoref]{lipics-v2019}

\usepackage[ruled,vlined,linesnumbered,commentsnumbered]{algorithm2e}
\usepackage{xcolor}
\usepackage{hyperref}

\title{A $2\sqrt{2k}$-approximation algorithm for \\ minimum power {\em k} edge disjoint {\em st}-paths}

\titlerunning{A $2\sqrt{2k}$-approximation algorithm for minimum power {\em k} edge disjoint {\em st}-paths}

\author{Zeev Nutov}{The Open University of Israel}{nutov@openu.ac.il}
{https://orcid.org/0000-0002-6629-3243}{}
\authorrunning{Zeev Nutov}
\Copyright{Zeev Nutov}

\nolinenumbers


\begin{document}

\maketitle


\begin{abstract}
In minimum power network design problems we are given 
an undirected graph $G=(V,E)$ with edge costs $\{c_e:e \in E\}$.  
The goal is to find an edge set $F \subseteq E$ that satisfies a prescribed property of minimum
power $p_c(F)=\sum_{v \in V} \max \{c_e: e \in F \mbox{ is incident to } v\}$.
In the {\sc Min-Power $k$ Edge Disjoint $st$-Paths} problem
$F$ should contain $k$ edge disjoint $st$-paths.
The problem admits a $k$-approximation algorithm, and it was an  open question 
to achieve an approximation ratio sublinear in $k$ even for unit costs. 
We give a $2\sqrt{2k}$-approximation algorithm for general costs.
\end{abstract}

\section{Introduction} \label{s:intro}

In network design problems one seeks a cheap subgraph 
that satisfies a prescribed property, often determined by pairwise connectivities or degree demands.  
A traditional setting is when each edge has a cost, and we want to minimize the cost of the subgraph.  
This setting does not capture many wireless networks scenarios, where a communication between two nodes 
depends on our ''investment'' in these nodes, like equipment and transmission energy,
and the cost is the sum of these ``investments''.
This motivates the type of problems we study here.
Specifically, we consider assigning transmission ranges to the nodes 
of a static ad hoc wireless network so as to minimize the total power consumed,  
under the constraint that the bidirectional network 
established by the transmission ranges satisfies prescribed properties. 

More formally, in {\bf minimum power network design problems} we are given 
an undirected (simple) graph $G=(V,E)$ with (non-negative) edge costs $\{c_e \geq 0:e \in E\}$.
The goal is to find an edge subset $F \subseteq E$ of minimum
total {\bf power} $p_c(F)=\sum_{v \in V} \max \{c_e: e \in \delta_F(v)\}$ that satisfies a prescribed property;
here $\delta_F(v)$ denotes the set of edges in $F$ incident to $v$, and a maximum taken over an empty set is assumed to be zero.
Equivalently, we seek an {\bf assignment}  $\{a_v \geq 0:v \in V\}$ to the nodes of minimum total value $\sum_{v \in V} a_v$, 
such that the {\bf activated edge set} $\{e=uv \in E: c_e \leq \min\{a_u,a_v\}\}$ satisfies the prescribed property. 
These problems were studied already in the late 90's, cf.~\cite{SRS,WNE,RM,KKKP,CH}, followed by many more. 
Min-power problems were also widely studied in directed graphs, 
usually under the assumption that to activate an edge one needs to assign power only to its tail, 
while heads are assigned power zero, cf.~\cite{KKKP,SM,N-powcov,N-sur}. 
The undirected case has an additional requirement -- we want the network to be bidirected,
to allow a bidirectional communication.

In the traditional edge-costs scenario, 
a fundamental problem in network design is the {\sc Shortest $st$-Path} problem.
A natural generalization and the simplest high connectivity network design problem 
is finding a set of $k$ disjoint $st$-paths of minimum edge cost. 
Here the paths may be edge disjoint -- the {\sc $k$ Edge Disjoint $st$-Paths} problem, 
or internally (node) disjoint -- the {\sc $k$ Disjoint $st$-Paths}  problem.
Both problems can be reduced to the 
{\sc Min-Cost $k$-Flow} problem, which has a polynomial time algorithm.   

Similarly, one of the most fundamental problems in the min-power setting is the {\sc Min-Power $st$-Path} problem. 
For this problem, an elegant linear time reduction  to the ordinary {\sc Shortest $st$-Path} problem is described by Althaus et al. \cite{ACMP}. 
Lando and Nutov \cite{LN} suggested a more general (but less time efficient) ''levels reduction'' 
that converts several min-power problems into problems with node costs.
A fundamental generalization is activating a set of $k$ edge disjoint or internally disjoint $st$-paths. 
Formally, the edge disjoint $st$-paths version is as follows. 

\begin{center} \fbox{\begin{minipage}{0.97\textwidth} \noindent
{\sc Min-Power $k$ Edge Disjoint $st$-Paths} ({\sc Min-Power $k$-EDP}) \\
{\em Input:} \ A (simple) graph $G=(V,E)$ with edge costs $\{c_e \geq 0:e \in E\}$, $s,t \in V$, 
and a positive integer $k$. \\
{\em Output:}   An edge set $F \subseteq E$ such that 
the graph $(V,F)$ contains $k$ edge disjoint $st$-paths 
of minimum power $p_c(F)=\sum_{v \in V} \max \{c_e: e \in \delta_F(v)\}$.
\end{minipage}}\end{center}

A related problem is the {\sc Node Weighted $k$ Edge Disjoint $st$-Paths} problem (a.k.a. {\sc Node-Weighted $k$-Flow}), 
where instead of edge costs we have node weights, and seek a min-weight subgraph that contains 
$k$ edge disjoint $st$-paths. For unit weights, 
this problem is equivalent to {\sc Min-Power $k$-EDP} with unit edge costs --
in both problems we seek a subgraph with minimum number of nodes that contains 
$k$ edge disjoint $st$-paths.
For {\sc Min-Power $k$-EDP} Lando and Nutov \cite{LN} obtained a $k$-approximation algorithm, improving over the easy ratio $2k$; 
they gave a polynomial time algorithm for the 
particular case of {\sc Min-Power $k$-EDP} when $G$ has a subgraph $G_0$ of cost $0$ that already contains $k-1$ edge disjoint 
$st$-paths, and we seek a min-power augmenting edge set $F$ to increase the number of paths by~$1$. 
On the other hand \cite{N-snw} shows that ratio $\rho$ for {\sc Min-Power $k$-EDP} with unit costs 
implies ratio $2\rho^2$ for the {\sc Densest $\ell$-Subgraph} problem,
that currently has best known ratio $O(n^{1/4+\epsilon})$ \cite{BCCF} 
and approximation threshold $\Omega\left(n^{1/poly (\log \log n)}\right)$ \cite{M}.
{\sc Min-Power $k$-EDP} was also studied earlier on directed graphs 
in many papers with similar results, cf.~\cite{SM,MMW,HKMN,LN}.
Speci\-fically, Hajiaghayi et al.~\cite{HKMN} showed that directed {\sc Min-Power $k$-EDP} is {\sc Label-Cover} hard,
while Maier, Mecke and Wagner \cite{MMW} showed that a natural algorithm for unit costs
has approximation ratio at least $2\sqrt{k}$.

Summarizing, even for unit costs, the best known approximation ratio for directed/undirected {\sc Min-Power $k$-EDP} was $k$, 
and the problem is unlikely to admit a polylogarithmic ratio.  
Open questions were to resolve the complexity status for $k=2$ \cite{AE,N-sur}, 
and to obtain an approximation ratio sublinear in $k$ (even just for unit activation costs) \cite{MMW,N-snw,N-sur}.  

The {\sc Min-Power $k$-EDP} problem is closely related to the {\sc Fixed Cost $k$-Flow} problem, 
in which each edge has a fixed capacity and cost, and the goal is to buy  
the cheapest set of edges to ensure an $st$-flow of value $k$. 
In this context, Hajiaghayi et al.~\cite{HKKN} studied the undirected {\sc Bipartite Fixed Cost $k$-Flow} problem,
where the graph $G \setminus \{s,t\}$ is bipartite, and $s$ is connected to one part and $t$ to the other part by edges of infinite
capacity and cost $1$, and the other edges have cost $0$ and capacity $1$. 
They gave for this problem an $O(\sqrt{k \ln k})$-approximation algorithm.
One can see that this problem is a very restricted case of {\sc Min-Power $k$-EDP} with $0,1$ costs. 
In fact, it was an open question to obtain approximation ratio sublinear in $k$ for {\sc Min-Power $k$-EDP}
even for $0,1$ costs \cite{MMW,N-snw,N-sur}. 
Our next result resolves this open question. 

\begin{theorem} \label{t2}
{\sc Min-Power $k$-EDP} admits a $2\sqrt{2k}$-approximation algorithm.
\end{theorem}

Theorem~\ref{t2} is based on the following combinatorial result that is of independent interest. 
Given a graph $G=(V,E)$ with edge costs $c_e$, 
we denote by $c(G) = c(E)=\sum_{e \in E} c_e$ the ordinary cost of $G$,
and by $p_c(G)=p_c(E)=\sum_{v \in V} \max\{c_e:e \in \delta_E(v)\}$ the power cost of $G$.
It is easy to see that $p_c(G) \leq 2c(G)$. 
Hajiaghayi et al.~\cite{HKMN} proved that $c(G) \leq \sqrt{|E|/2} \cdot p_c(G)$ for any graph $G$; this bound is tight.
We will improve over this bound for inclusion minimal simple graphs that contain $k$ edge disjoint $st$-paths.
Let us say that a graph $G$ is {\bf minimally $k$-$st$-edge-connected} if 
$G$ contains $k$ edge disjoint $st$-paths but no proper subgraph
of $G$ contains $k$ edge disjoint $st$-paths.

\begin{theorem} \label{t3}
Let $G=(V,E)$ be a minimally $k$-$st$-edge-connected  simple graph 
with non-negative edge costs $\{c_e:e \in E\}$.
Then $c(G) \leq \sqrt{2k} \cdot p_c(G)$. 
\end{theorem}

The $2\sqrt{2k}$-approximation algorithm will simply compute a minimum cost set of $k$ edge disjoint $st$-paths,
with edge costs $c_e$; this can be done in polynomial time using a min-cost $k$-flow algorithm.
The approximation ratio $2\sqrt{2k}$ follows from the Theorem~\ref{t3} bound,   
and the bound $p_c(G) \leq 2c(G)$; specifically, if 
$F^*$ is an optimal solution to {\sc Min-Power $k$-EDP} and 
$F$ is a min-cost set of $k$ edge disjoint paths, then  
$p_c(F) \leq 2c(F) \leq 2c(F^*) \leq 2\sqrt{2k} \cdot p_c(F^*)$.

\medskip

A large part of research in extremal graph theory deals with determining 
the maximal number of edges in multigraphs and simple graphs that are edge-minimal w.r.t. some specified property.
This question was widely studies for $k$-connected and $k$-edge-connected graphs. 
It is easy to show that a minimally $k$-edge-connected graph $G$ on $n$ nodes can be decomposed into $k$ forests 
and thus has at most $k(n -1)$ edges; a graph obtained from a spanning tree 
by replacing each edge by $k$ parallel edges shows that this bound is tight. 
Mader \cite{Mad} improved this bound for simple graphs -- 
in this case $G$ has at most $k(n-k)$ edges provided that $n \geq 3k -2$, 
and the complete bipartite graph $K_{k,n-k}$ shows that this bound is tight.
Mader \cite{Mad2} also showed that this bound holds for minimally $k$-connected graphs.
When $n < 3k-2$, a minimally $k$-edge-connected or $k$-connected simple graph 
has at most $\frac{1}{8}(n+k)^2$ edges, see \cite{Cai,BCW,Cai2}.
Similar bounds were established in \cite{N-ext} for graphs that are $k$-connected 
from a given root node to every other node (a.k.a. $k$-out-connected graphs).

To see that Theorem~\ref{t3} is indeed of interest, consider the case of unit costs.
Then $c(G)=|E|$ and $p_c(G)=V$. 
Thus already for this simple case we get the following fundamental extremal graph theory result, 
that to the best of our knowledge was not known before. 

\begin{corollary} \label{c:edp}
Let $G=(V,E)$ be a minimally $k$-$st$-edge-connected (directed or undirected) simple graph. 
Then $|E| \leq \sqrt{2k} \cdot |V|$.
\end{corollary}

In the next section we will show that this bound is asymptotically tight.

We briefly survey some results on more general activation network design problems.
Here we are given a graph $G=(V,E)$ with a pair of 
activation costs $\{c_e^u,c_e^v\}$ for each $uv$-edge $e \in E$;
the goal is to find an edge subset $F \subseteq E$ of minimum activation cost $\tau(F)=\sum_{v \in V} \max \{c_e^v: e \in \delta_F(v)\}$
that satisfies a prescribed property. 
This generic problem was introduced by Panigrahi \cite{P}, and it includes 
node costs problems (when for every $v \in V$, the costs $c_e^v$ are identical for all edges $e$ incident to $v$), 
min-power problems (when $c_e^u=c_e^v$ for each $uv$-edge $e \in E$), 
and several other problems that arise in wireless networks; see a survey in \cite{N-sur} on this type of problems. 
 
In the {\sc Activation $k$ Disjoint $st$-Paths} problem, the $k$ paths should be internally node disjoint. 
This problem admits an easy $2$-approximation algorithm. 
Based on an idea of Srinivas and Modiano \cite{SM}, Alqahtani and Erlebach \cite{AE} showed that 
approximation ratio $\rho$ for {\sc Activation $2$ Disjoint $st$-Paths} implies 
approximation ratio $\rho$ for {\sc Activation $2$ Edge Disjoint $st$-Paths}.
In \cite{AE} it is also claimed that {\sc Activation $2$ Disjoint Paths} admits a $1.5$-approximation algorithm, 
but the proof was found to contain an error -- see \cite{Thomas} where also a different $1.5$-approximation algorithm is given.
In another paper, Alqahtani and Erlebach \cite{AE-tw} showed 
that the problem is polynomially solvable on graphs with bounded treewidth.
However, it is a long standing open question whether {\sc Activation $2$ Disjoint $st$-Paths} is in P 
or is NP-hard on general graphs, even for power costs \cite{LN,AE,N-sur}.

\section{Proof of Theorem~\ref{t3} and a tight example} \label{s3}
 
Let $G=(V,E)$ be a simple minimally $k$-$st$-edge-connected graph.
Theorem~\ref{t3} says that then $c(G) \leq \sqrt{2k} \cdot p_c(G)$ for any edge costs $\{c_e:e \in E\}$, 
where $c(G)=\sum_{e \in E} c_e$ and $p_c(G)=\sum_{v \in V} \max\{c_e:e \in \delta(v)\}$
(recall that $\delta(v)$ denotes the set of edges in $E$  incident to $v$). 
As was mentioned in Corollary~\ref{c:edp}, in the case of uniform costs this reduces to $|E| \leq \sqrt{2k} \cdot |V|$. 
We will need a slight generalization of this bound to any subset $U$ of $V$, as follows. 

\begin{lemma} \label{l:edp}
Let $G=(V,E)$ be a simple minimally $k$-$st$-connected graph.
Then for any $U \subseteq V$,  $|E_U| \leq \sqrt{2k} \cdot |U|$,
where $E_U$ is the set of edges in $E$ with both ends in $U$.
\end{lemma}

This lemma will be proved in the next section. 
For now, we will use Lemma~\ref{l:edp} to prove Theorem~\ref{t3}, namely, that $c(G) \leq \sqrt{2k} \cdot p_c(G)$.
Note that $c(G) \leq \sum_{xy \in E} \min\{p_c(x),p_c(y)\}$, where $p_c(v)=\max\{c(e):e \in \delta(v)\}$.  
Thus, it is sufficient to prove that for any non-negative weights $\{p(v):v \in V\}$ on the nodes, 
the following holds:
\begin{equation} \label{e:p}
\sum_{xy \in E} \min\{p(x),p(y)\} \leq \sqrt{2k} \cdot \sum_{v \in V} p(v) \ .
\end{equation}

The proof of (\ref{e:p}) is by induction on the number $N$ of distinct $p(v)$ values. 
In the base case $N=1$, all weights $p(v)$ are equal, and w.l.o.g. $p(v)=1$ for all $v \in V$. 
Then (\ref{e:p}) reduces to $|E| \leq \sqrt{2k} \cdot |V|$, which is the case $U=V$ in Lemma~\ref{l:edp}. 

Assume that $N \geq 2$.
Let $U=\{u \in V:p(u)=\max_{v \in V} p(v)\}$ and let $E_U$ be the set of edges with both ends in $U$.
Let $\epsilon$ be the difference between the maximum weight $\max_{v \in V} p(v)$ and the second maximum weight. 
Let $p'$ be defined by $p'(v)=p(v)-\epsilon$ if $v \in U$ and $p'(v)=p(v)$ otherwise. 
Note that $|E_U| \leq \sqrt{2k} \cdot |U|$, by Lemma~\ref{l:edp}, and that $p'$ has exactly $N-1$ distinct values,
so by the induction hypothesis, (\ref{e:p}) holds for $p'$.  Thus we have:
\begin{eqnarray*}
\sum_{xy \in E} \min\{p(x),p(y)\} & =    & \sum_{xy \in E} \min\{p'(x),p'(y)\}+\epsilon |E_U|  \\
                                                     & \leq & \sqrt{2k} \sum_{v \in V} p'(v)+\epsilon \cdot \sqrt{2k}|U|  \\
																										 & =     & \sqrt{2k} \left(\sum_{v \in V} p'(v)+\epsilon|U|\right)=\sqrt{2k} \sum_{v \in V} p(v) \ .
\end{eqnarray*}
The first and last equalities are by the definition of $p'(v)$. 
The inequality is by the induction hypothesis and Lemma~\ref{l:edp}. 
This concludes the proof of Theorem~\ref{t3}, provided that we will prove Lemma~\ref{l:edp},
which we will do in the next section. 

\begin{figure}
\centering 
\includegraphics{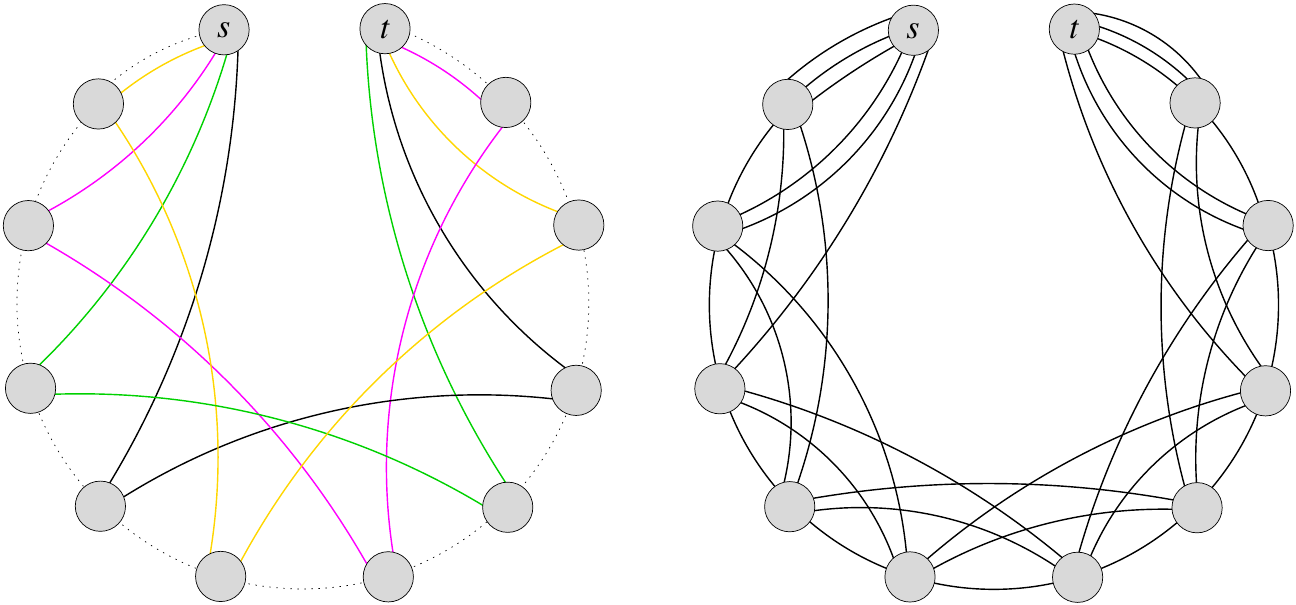}
\caption{$4$ disjoint $st$-paths in $F_4$ and the graph $G_3$.}
\label{f:tight}
\end{figure}

Before proving Lemma~\ref{l:edp}, let us give an example showing 
that for $U=V$ the bound in the lemma is asymptotically tight. 
Let $\ell,q$ be integer parameters, where $1 \leq \ell \leq q/2$. 
Let $P$ be an $st$-path with $q$ internal nodes.
The length of an edge that connects two nodes $u,v$ of $P$ is the distance between them in $P$.
Construct a graph $G_\ell$ as follows, see Fig.~\ref{f:tight}.

\begin{enumerate}
\item
Let $F_\ell$ 
consist of edges of length $\ell$ and edges incident to one of $s,t$ of length at most $\ell-1$. 
It is easy to see that $|F_\ell|=q+\ell$ and that $F_\ell$ is a union 
of $\ell$ edge disjoint (in fact, $\ell$ node internally disjoint) $st$-paths. 
\item
Let $G_\ell=(V,E_\ell)$ where $E_\ell$ is a disjoint union of $F_1,\ldots F_\ell$. 
One can verify that 
$G_\ell$ is minimally $k$-$st$-connected for $k=\ell(\ell+1)/2$ 
and that $|E_\ell|=q\ell+\ell(\ell+1)/2$.
\item
The graph $G_\ell$ is not simple, due to edges incident to $s$ and to $t$.
To make it a simple graph, subdivide $2k-\ell=\ell^2$ edges 
among the $2k$ edges incident to each of $s$ and $t$; 
this adds $\ell^2$ nodes and $\ell^2$ edges.
\end{enumerate}
Summarizing, the constructed graph $G_\ell$ has the following parameters:
\begin{eqnarray*}
k  & = & \ell(\ell+1)/2  \\ 
m & = & q\ell+\ell(\ell+1)/2+\ell^2 \\ 
n  & = & q+\ell^2+2
\end{eqnarray*}
For $\ell=\lceil 1/\epsilon\rceil$ and $q=\ell^4-\ell^2$, $m/n$ can be arbitrarily close to $\ell+1 > \sqrt{2k}$:
$$
\frac{m}{n} = 
\frac{q\ell+\ell(\ell+1)/2+\ell^2}{q+\ell^2+2} > 
\frac{q\ell}{q+\ell^2}=
\frac{\ell^2-1}{\ell} =
\frac{(\ell-1)(\ell+1)}{\ell} \geq (1-\epsilon)(\ell+1) \ .
$$

\section{Proof of Lemma~\ref{l:edp}} \label{s:lemma}

We prove Lemma~\ref{l:edp} for the {\em directed} graph obtained 
by orienting $k$ edge-disjoint $st$-paths in $G$ from $s$ to $t$;
note that this directed graph is also minimally $k$-$st$-edge-connected.    
For $R \subseteq V$ let $d_G^{\sf in}(R)$ and $d_G^{\sf out}(R)$ 
be the number of edges that enter and leave $R$, respectively.

\begin{lemma} \label{l:order}
Let $G=(V,E)$ be a minimally $k$-$st$-connected directed graph.
There exists an ordering $s=v_1, \ldots, v_n=t$ of \ $V$ such that for every $i \le n-1$, 
$C_i=\{v_1, \ldots v_i\}$ is a minimum $st$-cut and no edge enters $C_i$. 
\end{lemma}
\begin{proof}
For $n=2$ the lemma is trivial so assume that $n \geq 3$. 
We will show that then $G$ has a node $z \in V \setminus\{ s,t\}$ such that the following holds: \\
{\em $s$ is the tail of every edge that enters $z$ and 
the graph obtained from $G$ by contracting $z$ into $s$ and deleting loops 
is also minimally $k$-$st$-connected.} \\
We then define the order recursively as follows. 
We let $G_1=G$ and $v_1=s$. 
For $i=2, \ldots,n-1$ we let $v_i$ be a node $z$ as above of $G_{i-1}$, and $G_i$ is 
obtained from $G_{i-1}$ by contracting $v_i$ into $s$.
It is easy to see that the order is as required, so we only need to prove existence of $z$ as above.

Note that $d_G^{\sf in}(v)=d_G^{\sf out}(v) \geq 1$ for every $v \neq s,t$,
while $d_G^{\sf out}(s)=d_G^{\sf in}(t)=k$ and $d_G^{\sf in}(s)=d_G^{\sf out}(t)=k$.
By considering $G$ as a flow network with unit capacity and unit flow on every edge,
we get the following.
\begin{itemize}
\item
If $R$ is a minimum $st$-cut of $G$ (so $d_G^{\sf out}(R)=k$) then $d_G^{\sf in}(R)=0$;
this is so since for any $st$-cut $R$ the flow value equals to $d_G^{\sf out}(R)-d_G^{\sf in}(R)$.
\item
$G$ is acyclic, since removing a cycle does not affect the flow value, 
hence existence of a cycle contradicts the minimality of $G$.
\end{itemize}

The graph $G \setminus \{s,t\}$ is also acyclic and thus has a node $z$ such that no edge enters $z$,
so $s$ is the tail of every edge that enters $z$.
Let $G'=(V',E')$ be obtained from $G$ by contracting $z$ into $s$. 
It is easy to see that $G'$ is $k$-$st$-connected, so we only have to prove minimality.
For that it is sufficient to prove that for any $e \in E'$ (so $e$ can be any edge of $G$ that is not an $sz$-edge)
there exists a minimum $st$-cut $R$ of $G$ with $e \in \delta^{\sf out}_G(R)$ such that $z \in R$. 

By the minimality of $G$,  there exists a minimum $st$-cut $R$ with $e \in \delta^{\sf out}_G(R)$.
If $z \in R$ then we are done. 
Otherwise, we claim that $R \cup \{z\}$ is also a minimum $st$-cut of $G$.
Note that all edges entering $z$ are in $\delta_G^{\sf out}(R)$ (since their tail is $s$) and 
all edges leaving $z$ have head in $V \setminus R$ (since $G$ has no edges entering $R$). Thus 
$
d^{\sf out}_G(R \cup \{z\})=d_G^{\sf out}(R)-d_G^{\sf in}(z)+d_G^{\sf out}(z)=k
$,
concluding the proof. 
\end{proof}

\begin{lemma} \label{l:bound}
Let $v_1, \ldots, v_n$  be an ordering of nodes of a directed simple graph $G=(V,E)$ 
and let  $C_i=\{v_1, \ldots v_i\}$ for $i=1, \ldots, n-1$. 
If $d_G^{\sf out}(C_i) \leq k$ and $d_G^{\sf in}(C_i) =0$ for all $i$ then $|E| \leq \sqrt{2k} n$. 
\end{lemma}
\begin{proof}
Let us define the {\bf length} $l(e)$ of an edge $e=v_iv_j$ to be $j-i$. 
Note that  the total length of $E$ is bounded by 
$$
l(E)=\sum_{e \in E} l(e)=\sum_{i=1}^{n-1} d_G^{\sf out}(C_i) \leq k(n-1) < kn \ .
$$
We claim that any simple edge set $E$ with $l(E) \leq kn$ has at most $\sqrt{2k} n$ edges.
Clearly, a maximum size edge set is obtained by picking the shortest edges first, 
say an edge set $E_\ell$ of all edges of length at most $\ell$ and some $i \leq n-(\ell+2)$ edges of length $\ell+1$.
Note that there are exactly $n-j$ edges of length $j$ of total length $j(n-j)$. 
Thus the size and the total length of $E_\ell$ are:
\begin{eqnarray*}
|E_\ell|  & = & \sum_{j=1}^{\ell} (n-j)= n\ell - \frac{\ell(\ell+1)}{2}= \ell\left(n-\frac{\ell+1}{2}\right) \\
l(E_\ell) & = & \sum_{j=1}^\ell j(n-j) =  
n \frac{\ell(\ell+1)}{2} - \frac{\ell(\ell+1)(2\ell+1)}{6} =
\frac{\ell(\ell+1)}{2} \left(n-\frac{2\ell+1}{3} \right) 
\end{eqnarray*}
So $\ell$ is the largest integer such that $l(E_\ell) \leq kn$ 
and $i$ is the largest integer such that $i(\ell+1) \leq kn-l(E_\ell)$. 
The maximum possible number of edges is $m^*=\ell\left(n-\frac{\ell+1}{2}\right)+i$, therefore:
$$
\frac{|E|}{n} \leq \frac{m^*}{n}=\ell\left(1-\frac{\ell+1}{2n}\right)+\frac{i}{n} \ .
$$
Since $l(E_\ell) +i(\ell+1) \leq kn$ we have:
$$
\ell(\ell+1)\left(1-\frac{2\ell+1}{3n}\right) +\frac{2i(\ell+1)}{n} \leq 2k
$$
Thus to prove that $(m^*/n)^2 \leq 2k$ it is sufficient to prove that:
$$
\left(\ell\left(1-\frac{\ell+1}{2n}\right)+\frac{i}{n}\right)^2 \leq \ell(\ell+1)\left(1-\frac{2\ell+1}{3n} \right)+\frac{i}{n} \cdot 2(\ell+1) \ .
$$
Easy computations show that (recall that $i < n$):
$$
\left(1-\frac{\ell+1}{2n}\right)^2  < 1-\frac{2\ell+1}{3n} \ \ \ \ \ \ \ \ 
\ell\left(1-\frac{\ell+1}{2n}\right) + \frac{i}{2n} < \ell+1 \ .
$$
Thus we get:
\begin{eqnarray*}
\left(\ell\left(1-\frac{\ell+1}{2n}\right)+\frac{i}{n}\right)^2 
&   =   & 
\ell^2\left(1-\frac{\ell+1}{2n}\right)^2+\frac{i}{n}\left(2\ell\left(1-\frac{\ell+1}{2n}\right) + \frac{i}{n}\right) \\
& < &
\ell(\ell+1)\left(1-\frac{2\ell+1}{3n} \right)+\frac{i}{n} \cdot 2(\ell+1) \ ,
\end{eqnarray*}
as required.
\end{proof}

Let now $G$ and $v_1, \ldots,v_n$ be as in Lemma~\ref{l:bound}. 
Note that for any node subset $U \subseteq V$, 
the subsequence $u_1, \ldots,u_{|U|}$ of $v_1, \ldots,v_n$ of the nodes in $U$
and the subgraph $G[U]=(U,E_U)$ induced by $U$ 
satisfy the conditions of Lemma~\ref{l:bound},
implying that $|E_U| \leq \sqrt{2k} |U|$. 

\medskip

This concludes the proof of Lemma~\ref{l:edp}, and thus also the proof of Theorem~\ref{t3} is complete. 

\medskip \medskip

\noindent {\bf Acknowledgment.}
I thank an anonymous referee for noticing an error in one of the versions of this paper.
I also thank several anonymous referees for useful comments that helped to improve the presentation.


\begin{thebibliography}{10}

\bibitem{AE}
H.~M. Alqahtani and T.~Erlebach.
\newblock Approximation algorithms for disjoint $st$-paths with minimum
  activation cost.
\newblock In {\em Algorithms and Complexity, 8th International Conference
  (CIAC)}, pages 1--12, 2013.

\bibitem{AE-tw}
H.~M. Alqahtani and T.~Erlebach.
\newblock Minimum activation cost node-disjoint paths in graphs with bounded
  treewidth.
\newblock In {\em 40th Conference on Current Trends in Theory and Practice of
  Computer Science (SOFSEM)}, pages 65--76, 2014.

\bibitem{ACMP}
E.~Althaus, G.~Calinescu, I.~Mandoiu, S.~Prasad, N.~Tchervenski, and
  A.~Zelikovsky.
\newblock Power efficient range assignment for symmetric connectivity in static
  ad-hoc wireless networks.
\newblock {\em Wireless Networks}, 12(3):287--299, 2006.

\bibitem{BCCF}
A.~Bhaskara, M.~Charikar, E.~Chlamtac, U.~Feige, and A.~Vijayaraghavan.
\newblock Detecting high log-densities: an {$O(n^{1/4})$} approximation for
  densest $k$-subgraph.
\newblock In {\em 42nd ACM Symposium on Theory of Computing (STOC)}, pages
  201--210, 2010.

\bibitem{BCW}
K~Budayasa, Lou Caccetta, and Western Australia.
\newblock On critically $k$-edge-connected graphs.
\newblock {\em Australas. J Comb.}, 2:101--110, 1990.

\bibitem{Cai}
M.~C. Cai.
\newblock Minimally $k$-connected graphs of low order and maximal size.
\newblock {\em Discrete Math.}, 41(3):229--234, 1982.

\bibitem{Cai2}
M.~C. Cai.
\newblock The maximal size of graphs with at most k edge-disjoint paths
  connecting any two adjacent vertices.
\newblock {\em Discrete Math.}, 85(1):43--52, 1990.

\bibitem{CH}
W.~Chen and N.~Huang.
\newblock The strongly connecting problem on multi-hop packet radio networks.
\newblock {\em IEEE Trans. on Commun.}, 37(3):293--295, 1989.

\bibitem{HKMN}
M.~Hajiaghayi, G.~Kortsarz, V.~Mirrokni, and Z.~Nutov.
\newblock Power optimization for connectivity problems.
\newblock {\em Math. Progr.}, 110(1):195--208, 2007.

\bibitem{HKKN}
M-.T. Hajiaghayi, R.~Khandekar, G.~Kortsarz, and Z.~Nutov.
\newblock On fixed cost $k$-flow problems.
\newblock {\em Theory of Computing Systems}, 58(1):4--18, 2016.

\bibitem{KKKP}
L.~M. Kirousis, E.~Kranakis, D.~Krizanc, and A.~Pelc.
\newblock Power consumption in packet radio networks.
\newblock {\em Theor. Comput. Sci.}, 243(1-2):289--305, 2000.

\bibitem{LN}
Y.~Lando and Z.~Nutov.
\newblock On minimum power connectivity problems.
\newblock {\em J. Discrete Algorithms}, 8(2):164--173, 2010.

\bibitem{Mad}
W.~Mader.
\newblock Minimale $n$-fach kantenzusammenh\"{a}ngende graphen.
\newblock {\em Mathematische Annalen}, 191:21--28, 1971.

\bibitem{Mad2}
W.~Mader.
\newblock \"{U}ber minimal $n$-fach zusammenh\"{a}ngende, unendliche graphen
  und ein extremalproblem.
\newblock {\em Arch. Math. (Basel)}, 23:553--560, 1972.

\bibitem{MMW}
M.~Maier, S.~Mecke, and D.~Wagner.
\newblock Algorithmic aspects of minimum energy edge-disjoint paths in wireless
  networks.
\newblock In {\em 33rd Conference on Current Trends in Theory and Practice of
  Computer Science (SOFSEM)}, pages 410--421, 2007.

\bibitem{M}
P.~Manurangsi.
\newblock Almost-polynomial ratio {ETH}-hardness of approximating densest
  $k$-subgraph.
\newblock In {\em 49th Symposium on Theory of Computing (STOC)}, pages
  954--961, 2017.

\bibitem{N-powcov}
Z.~Nutov.
\newblock Approximating minimum power covers of intersecting families and
  directed edge-connectivity problems.
\newblock {\em Theoretical Computer Science}, 411(26-28):2502--2512, 2010.

\bibitem{N-snw}
Z.~Nutov.
\newblock Approximating {S}teiner networks with node-weights.
\newblock {\em SIAM J. Computing}, 39(7):3001--3022, 2010.

\bibitem{N-sur}
Z.~Nutov.
\newblock Activation network design problems.
\newblock In T.~F. Gonzalez, editor, {\em Handbook on Approximation Algorithms
  and Metaheuristics, Second Edition}, volume~2, chapter~15. Chapman \&
  Hall/CRC, 2018.

\bibitem{Thomas}
Z.~Nutov and D.~Kahba.
\newblock A 1.5-approximation algorithms for activating 2 disjoint $st$-paths.
\newblock {\em CoRR}, abs/2307.12646, 2023.
\newblock URL: \url{https://doi.org/10.48550/arXiv.2307.12646}, \href
  {https://arxiv.org/abs/2307.12646} {\path{arXiv:2307.12646}}.

\bibitem{N-ext}
Zeev Nutov.
\newblock On extremal $k$-outconnected graphs.
\newblock {\em Discrete Math.}, 308(12):2533--2543, 2008.

\bibitem{P}
D.~Panigrahi.
\newblock Survivable network design problems in wireless networks.
\newblock In {\em 22nd Symp. on Discr. Algorithms (SODA)}, pages 1014--1027,
  2011.

\bibitem{RM}
V.~Rodoplu and T.~H. Meng.
\newblock Minimum energy mobile wireless networks.
\newblock In {\em IEEE International Conference on Communications (ICC)}, pages
  1633--1639, 1998.

\bibitem{SRS}
S.~Singh, C.~S. Raghavendra, and J.~Stepanek.
\newblock Power-aware broadcasting in mobile ad hoc networks.
\newblock In {\em Proceedings of IEEE PIMRC}, 1999.

\bibitem{SM}
A.~Srinivas and E.~H. Modiano.
\newblock Finding minimum energy disjoint paths in wireless ad-hoc networks.
\newblock {\em Wireless Networks}, 11(4):401--417, 2005.

\bibitem{WNE}
J~E. Wieselthier, G.~D. Nguyen, and A.~Ephremides.
\newblock On the construction of energy-efficient broadcast and multicast trees
  in wireless networks.
\newblock In {\em Proc. IEEE INFOCOM}, pages 585--594, 2000.

\end{thebibliography}

\end{document}